\documentclass[12pt,reqno,a4paper]{article}
\usepackage{amsmath,amssymb,dsfont,graphicx,latexsym,
epsf,cancel,epic,eepic,amsthm,tikz-cd,mathtools}
\usepackage[colorlinks=false]{hyperref}
\usepackage{mathpazo}
\usepackage{tikz-cd}

\newtheorem{theorem}{Theorem}
\newtheorem{proposition}{Proposition}

\newcommand{\bbC}{{\mathbb C}}
\newcommand{\bbP}{{\mathbb P}}

\newcommand{\cC}{{\mathcal C}}

\newcommand{\cP}{{\mathcal P}}
\newcommand{\cQ}{{\mathcal Q}}

\title{Invariant volume form for 3D QRT maps}
\author{Jaume Alonso \and Yuri B.\ Suris}

\begin{document}

\maketitle

%
%
%

\begin{abstract}
Recently, we proposed a three-dimensional generalization of QRT maps. These novel maps can be associated with pairs of pencils of quadrics in $\mathbb P^3$. By construction, these maps have two rational integrals (parameters of both pencils). In the present paper, we find an invariant volume form for these maps, thus finally establishing their integrability.
\end{abstract}

\section{Introduction}

The so called QRT maps belong to the most celebrated examples of two-dimensional discrete integrable systems \cite{QRT1, QRT2, QRT book}. Recently, we introduced their three-dimensional generalization \cite{ASW}. The 3D QRT maps turned out to be instrumental for finding novel integrable Kahan-Hirota-Kimura type discretizations \cite{ASW}, as well as for a novel approach to discrete Painlev\'e equations \cite{ASW1, AS2}, but several inherent queries about these maps remained unanswered to this day. In the present paper, we resolve one of such queries.

By construction (to be quickly recalled below), any 2D QRT map has one integral of motion, while any 3D QRT map has two integrals of motion. This is often interpreted as {\em integrability} of these maps, however such a conclusion would be premature. Indeed, established notions of integrability require more than just the existence of a certain number of integrals of motion.

For a symplectic map on a symplectic manifold of (necessarily  even) dimension $2n$, the {\em Liouville-Arnold integrability} requires the existence of $n$ functionally independent integrals with pairwise vanishing Poisson brackets. For $n=1$, the latter requirement becomes void, as only one integral of motion is needed, but one still needs the symplecticity of the map, which in this case is equivalent to the existence of an invariant non-degenerate two-form. As a matter of fact, QRT maps do possess invariant two-forms, but this fact is less well-known than the existence of an integral of motion.  This has been established in the PhD thesis by John Roberts (``R'' in QRT) \cite{R}, while the first published proof is in \cite{IR}. The main benefit of the Liouville-Arnold integrabilty in dimension 2 is the statement that on any compact invariant curve (level set of an integral of motion), the dynamics are linear in a suitable affine structure.

Clearly, the notion of Liouville-Arnold integrability is not applicable to maps in odd dimensions. The following can be proposed as its replacement for dimension 3: a three-dimensional map is called integrable, if it possesses an invariant volume form and two functionally independent integrals $I_1$, $I_2$. Indeed, in this case the map can be restricted to level sets of $I_1$, and such a restriction will possess an integral of motion induced by $I_2$, and according to \cite{BHQ}, an invariant two-form. Thus, the dynamics on any compact invariant curve (common level set of $I_1$, $I_2$) are again linear in a suitable affine structure.

Thus, to declare a three-dimensional map with two integrals of motion \emph{integrable}, one needs to establish that it possesses an invariant volume form. Doing this for 3D QRT maps is exactly the goal of this paper.

It is organised as follows. In sections \ref{sect 2D QRT}, \ref{sect 3D QRT} we give a quick introduction into the geometric construction and main properties of two-dimensional QRT maps and their three-dimensional generalizations, respectively. In section \ref{sect adapted coords}, we describe the so called pencil-adapted coordinates for pencils of quadrics in $\bbP^3$, which turn out to be the main technical tool for our purpose here. In section \ref{sect main th}, we formulate the main theorem about the invariant volume form for 3D QRT maps, and discuss how its apparent lack of symmetry can be remedied. In section \ref{sect main proof}, the idea of the proof is presented, based on an important technical result (Proposition \ref{prop Q infty}). The proof of the latter is relegated to section \ref{sect proof prop 3}, after the statement is illustrated with some examples in section \ref{sect examples}. The paper is concluded by a short outlook in section \ref{sect conclusions}.

\section{2D QRT maps}
\label{sect 2D QRT}

To quickly introduce QRT maps, consider a \emph{pencil of biquadratic curves}
$$
\cC_\mu=\Big\{(x,y)\in \bbC^2: C_\mu(x,y):=C_0(x,y)-\mu C_\infty(x,y)=0\Big\},
$$
where $C_0,C_\infty$ are two polynomials of bidegree (2,2). The \emph{base set} $\mathcal B$ of the pencil is defined as the set of points through which all curves of the pencil pass or, equivalently, as the intersection $\{C_0(x,y)=0\}\cap\{C_\infty(x,y)=0\}$. 
Through any point $(x_0,y_0)\not\in \mathcal B$, there passes exactly one curve of the pencil, defined by $\mu=\mu(x_0,y_0)=C_0(x_0,y_0)/C_\infty(x_0,y_0)$. Actually, we consider this pencil in a compactification $\bbP^1\times \bbP^1$ of $\bbC^2$. Then, $\mathcal B$ consists of eight base points, counted with multiplicity.

One defines the \emph{vertical switch} $i_1$ and the \emph{horizontal switch} $i_2$ as follows. For a given point $(x_0,y_0)\in\bbP^1\times \bbP^1 \setminus \mathcal B$, determine $\mu=\mu(x_0,y_0)$ as above. Then the vertical line $\{x=x_0\}$ intersects $\cC_\mu$ at exactly one further point $(x_0,y_1)$ which is defined to be $i_1(x_0,y_0)$; similarly, the horizontal line $\{y=y_0\}$ intersects $\cC_\mu$ at exactly one further point $(x_1,y_0)$ which is defined to be $i_2(x_0,y_0)$. The QRT map is defined as
$$
f=i_1\circ i_2.
$$
Each of the maps $i_1$, $i_2$ is a birational involution on $\bbP^1\times \bbP^1$ with indeterminacy set $\mathcal B$. Likewise, the QRT map $f$ is a (dynamically nontrivial) birational map on $\bbP^1\times \bbP^1$, having $\mu(x,y)=C_0(x,y)/C_\infty(x,y)$ as an integral of motion. A generic fiber $C_\mu$ is an elliptic curve, and $f$ acts on it as a shift with respect to the corresponding addition law. 

\begin{proposition}\label{2D QRT inv measure}
The map $f$ has an invariant area form $dx\wedge dy/C_{\infty}(x,y)$ with the density $C_{\infty}(x,y)$. This can be expressed as
$$
\det\frac{\partial(\widetilde x,\widetilde y)}{\partial(x,y)}=\frac{C_{\infty}(\widetilde x,\widetilde y)}{C_{\infty}(x,y)},
$$
where we denote $(\widetilde x,\widetilde y)=f(x,y)$.
\end{proposition}

{\bf Remark.} In the statement of Proposition \ref{2D QRT inv measure}, one can replace $C_\infty$ by $C_{\mu_0}$ for any fixed $\mu_0\in\bbP^1$. Indeed, the density of an invariant area form is defined up to the multiplication with an integral of motion, and the quotient $C_{\mu_0}(x,y)/C_{\infty}(x,y)$ is an integral of motion of $f$ (for any fixed $\mu_0$). As pointed out above, this proposition is due to  \cite{R,IR}. For the reader's convenience, we reproduce its proof in the Appendix.

\section{3D QRT maps}
\label{sect 3D QRT}

As introduced in \cite{ASW}, a 3D QRT map is defined by two {\em pencils of quadrics} $\{\cQ_\lambda\}_{\lambda\in\bbP^1}$ and $\{\cP_\mu\}_{\mu\in\bbP^1}$ in $\bbP^3$,
\begin{eqnarray*}
\cQ_\lambda & = & \big\{X\in\bbP^3: Q_\lambda(X):=Q_0(X)-\lambda Q_\infty(X)=0\big\},\\
\cP_\mu & = & \big\{X\in\bbP^3: P_\mu(X):=P_0(X)-\mu P_\infty(X)=0\big\},
\end{eqnarray*} 
where $Q_0,Q_\infty,P_0,P_\infty$ are four homogeneous polynomials of $X=(X_1,X_2,X_3,X_4)$ of degree 2. The roles of the two pencils in the construction are different. 

We assume that the quadric $\cQ_0$ is non-degenerate; without loss of generality, we always take $Q_0(X)=X_1X_2-X_3X_4$. Denote by $M_0,M_\infty\in {\rm Sym}_{4\times 4}(\bbC)$ symmetric matrices of the quadratic forms $Q_0,Q_\infty$, and set $M_\lambda=M_0-\lambda M_\infty$, 
\begin{equation}\label{Delta}
\Delta(\lambda)=\det(M_\lambda)=\det(M_0-\lambda M_\infty)
\end{equation}
The characteristic polynomial $\Delta(\lambda)$ of the pencil $\{\cQ_\lambda\}$ is considered to have degree 4, so that $\lambda=\infty$ counts as a zero of $\Delta(\lambda)$ exactly if its leading coefficient (i.e., the coefficient by $\lambda^4$) vanishes. We set
\begin{equation}\label{Sing}
{\rm Sing}(\cQ)=\{\lambda\in\bbP^1: \cQ_\lambda\;\;{\rm is\;\; singular}\}=\{\lambda\in\bbP^1: \Delta(\lambda)=0\}.
\end{equation}
  
For any $X\in \bbP^3$ not belonging to the base set $\cQ_0\cap\cQ_\infty$ of the first pencil $\{\cQ_\lambda\}$, let $\lambda=\lambda(X)$ be defined as the unique value of the pencil parameter for which $X\in \cQ_\lambda$. Suppose that $\cQ_\lambda$ is non-degenerate. Then it admits two rulings such that any two lines of one ruling are skew and any line of one ruling intersects any line of the other ruling. Through each point $X\in \cQ_\lambda$ there pass two straight lines, one of each of the two rulings, let us call them $\ell_1(X)$ and $\ell_2(X)$. 

Suppose that $X$ does not belong to the base set of the second pencil $\{\cP_\mu\}$ either,  so that one can define $\mu=\mu(X)$ as the unique value of the pencil parameter for which $X\in \cP_\mu$. Now we are in the position to define the 3D QRT involutions $i_1, i_2$ on $\bbP^3$:
\begin{itemize}
\item $i_1(X)$ is the second intersection point of the generator $\ell_1(X)$ of $\cQ_\lambda$ with $\cP_\mu$, 
\item $i_2(X)$ is the second intersection point of the generator $\ell_2(X)$ of $\cQ_\lambda$ with $\cP_\mu$. 
\end{itemize}
The 3D QRT map on $\bbP^3$ is defined as $f=i_1\circ i_2$. By construction, $f$ leaves all quadrics of both pencils invariant. In other words, both rational functions $\lambda=Q_0(X)/Q_\infty(X)$ and $\mu=P_0(X)/P_\infty(X)$ are integrals of motion for $f$.

The main problem with this definition is that the dependence of generators $\ell_1(X)$, $\ell_2(X)$ on the point $X$ can be non-rational.  This issue will be clarified in the following section.

\section{Pencil-adapted coordinates}
\label{sect adapted coords}

One can classify pencils of quadrics in $\bbP^3$, containing at least one non-degenerate quadric, modulo complex congruence transformations, see, e.g., \cite{CA, OSG}. We briefly remind the main results, they can be also looked up in \cite{ASW1, AS2}. 

There are thirteen classes of pencils of quadrics, which can be distinguished either geometrically, by the type of the base curve of the pencil, or algebraically, by the structure of the system of elementary divisors of $M_\lambda$, encoded in the so called \emph{Segre symbols}. Elementary divisors are powers of $\lambda-\lambda_k$ for $\lambda_k\in\,{\rm Sing}(\cQ)$. The product of all elementary divisors is the characteristic polynomial $\Delta(\lambda)$.

Assuming, without loss of generality, that $Q_0(X)=X_1X_2-X_3X_4$, consider the problem of reducing $Q_\lambda(X)$ to the normal form $Q_0(X)$. Thus, we look for a linear projective change of variables $X=A_\lambda Y$ such that\begin{equation}\label{A norm}
Q_\lambda(A_\lambda Y)=Q_0(Y), \quad {\rm or}\quad A_\lambda^{\rm T}M_\lambda A_\lambda=M_0=\begin{pmatrix}  0 & 1 & 0 & 0 \\ 1 & 0 & 0 & 0  \\
0 & 0 & 0 & -1 \\ 0 & 0 & -1 & 0 \end{pmatrix}.
\end{equation}
A standard result from linear algebra says:
\begin{itemize}
\item[] \emph{The normalizing matrix $A_\lambda$ can be chosen as a rational function of $\lambda$ and of $\sqrt{\Delta(\lambda)}$. }
\end{itemize}
In particular, if $\Delta(\lambda)$ is a complete square (which happens for seven out of thirteen classes), then $A_\lambda$ is a rational function of $\lambda$. In the remaining six classes, it is more natural to consider this matrix as a meromorphic function on the Riemann surface $\mathcal R$ of $\sqrt{\Delta(\lambda)}$. This Riemann surface is a double cover of $\widehat{\mathbb C}$ branched at two or at four points. By the uniformization theorem, its universal cover is $\mathbb C$. We will denote the uniformizing variable $\nu\in\mathbb C$, so that the maps $\nu\mapsto \lambda$ and $\nu\mapsto \sqrt{\Delta(\lambda)}$ are holomorphic. In this case equation \eqref{A norm} becomes
\begin{equation}\label{A norm branched}
Q_{\lambda(\nu)}(A_{\nu} Y)=Q_0(Y), \quad {\rm or}\quad A_{\nu}^{\rm T}M_{\lambda(\nu)} A_{\nu}=M_0.
\end{equation}
For simplicity, we will sometimes slightly abuse notation, sticking to formula \eqref{A norm} also in the branched case.
\smallskip

The main application of the normalizing transformation $A_\lambda$ are so called \emph{pencil-adapted coordinates}. Namely, the formulas 
\begin{equation}\label{phi lambda}
\begin{pmatrix} X_1 \\ X_2 \\ X_3 \\ X_4 \end{pmatrix}=\phi_\lambda(x,y):=A_\lambda\begin{pmatrix} x \\ y \\ xy \\ 1 \end{pmatrix},
\end{equation} 
(in the rational case), resp. 
\begin{equation}\label{phi nu}
\begin{pmatrix} X_1 \\ X_2 \\ X_3 \\ X_4 \end{pmatrix}=\phi_\nu(x,y):=A_\nu\begin{pmatrix} x \\ y \\ xy \\ 1 \end{pmatrix},
\end{equation} 
(in the branched case) give a parametrization of $\cQ_\lambda$ (resp. of $\cQ_{\lambda(\nu)}$) by $(x,y)\in\bbP^1\times\bbP^1$, such that the generators $\ell_1$ of $\cQ_\lambda$ correspond to $x={\rm const}$, while generators $\ell_2$ correspond to $y={\rm const}$.  In the first case, generators $\ell_1(X)$, $\ell_2(X)$ are rational functions on $\mathbb P^3$, and  the 3D QRT maps are birational. In the second case, generators become well-defined rational functions on the variety $\mathcal X$ which is a branched double covering of $\bbP^3$, whose ramification locus is the union of the singular quadrics $\cQ_{\lambda_i}$, where $\lambda_i$ are the branch points of $\mathcal R$. Interchanging two sheets of the covering corresponds to interchanging two families of generators $\ell_1$, $\ell_2$. Likewise, the 3D QRT maps become single-valued on $\mathcal X$. Anyway, the preservation of a measure is a local property, which we consider and establish away from the degenerate quadrics. 
\smallskip

We will now illustrate the concepts of the normalizing transformation and the pencil-adapted coordinates by two examples, chosen among the thirteen classes of pencils in a pretty random way. For other classes, the reader can consult \cite{ASW1, AS2}. For the most general case ($\Delta(\lambda)$ with four simple roots, the base curve of the pencil being a non-degenerate spatial curve of degree 4 and of genus 1), the pencil-adapted coordinates will be given in a separate publication.
\smallskip

\emph{Example 1: pencil of quadrics through a twisted cubic and one of its tangent lines:}
\begin{equation}
\cQ_\lambda=\big\{X_1X_2-X_3X_4-\lambda(X_1^2-X_2X_4)=0\big\},  \quad 
M_\lambda=\begin{pmatrix}    -2\lambda & 1 & 0 & 0 \\ 1 & 0 & 0 & \lambda \\ 0 & 0 & 0 & -1 \\ 0 & \lambda & -1 & 0
\end{pmatrix}.
\end{equation}
Here $\Delta(\lambda)=1$, a complete square, with ${\rm Sing}(\cQ)=\{\infty\}$. The base curve of this pencil is:
\begin{equation}\label{case vi base curve}
\cQ_0\cap \cQ_\infty = \big\{[x:x^2:x^3:1]: x\in\bbP^1\big\}\cup\{X_1=X_4=0\} 
\end{equation} 
(the union of a twisted cubic and a tangent line to it). In this case, one finds:
$$
A_\lambda=\begin{pmatrix} 1 & 0 & 0 & 0 \\ \lambda & 1 & 0 & 0 \\ \lambda^2 & \lambda & 1 & 0 \\ 0 & 0 & 0 & 1
\end{pmatrix}.
$$
Parametrization of $\cQ_\lambda$ by the pencil-adapted coordinates: 
\begin{equation}
\phi_\lambda(x,y) =\begin{pmatrix} x \\ y+\lambda x\\  xy+\lambda y+\lambda^2 x \\ 1\end{pmatrix}.
 \end{equation}

\smallskip

\emph{Example 2: pencil of quadrics through a cuspidal quartic curve:}
\begin{equation}
\cQ_\lambda=\big\{X_1X_2-X_3X_4-\lambda Q_\infty(X)=0\big\}\;\;{\rm with}\;\;
Q_\infty(X)=\tfrac{1}{2}(X_1-X_2)^2-(X_1+X_2)X_4,
\end{equation}
$$
M_\lambda=\begin{pmatrix} -\lambda & 1+\lambda & 0 & \lambda \\ 1+\lambda & -\lambda & 0 & \lambda \\
0 & 0 & 0 & -1 \\ \lambda & \lambda & -1 & 0
\end{pmatrix}
$$
For this pencil $\Delta(\lambda)=2\lambda+1$, not a complete square, with ${\rm Sing}(\cQ)=\{-\tfrac{1}{2},\infty\}$. The base curve $\cQ_0\cap \cQ_\infty$ is a cuspidal space curve of degree 4, with the cusp at $[0:0:1:0]$. A uniformizing variable for $\sqrt{\Delta(\lambda)}$ can be chosen as
\begin{equation}
\nu=\sqrt{\Delta(\lambda)},\quad \lambda(\nu)=(\nu^2-1)/2.
\end{equation}
One finds:
$$
A_\nu=\begin{pmatrix} 
 \tfrac{1}{2\nu}(\nu+1) &  \tfrac{1}{2\nu}(\nu-1) & 0 & 0\\[0.2cm]
 \tfrac{1}{2\nu}(\nu-1) &  \tfrac{1}{2\nu}(\nu+1) & 0 & 0\\[0.2cm]
 \tfrac{1}{2}(\nu^2-1) &   \tfrac{1}{2}(\nu^2-1) & 1 & 0\\[0.2cm]
 0 & 0 & 0 & 1\end{pmatrix}.
 $$
Parametrization of $\cQ_{\lambda(\nu)}$ by the pencil-adapted coordinates: 
\begin{equation}
\phi_{\nu}(x,y)
 =\begin{pmatrix} \frac{1}{2\nu}\big((\nu+1)x+(\nu-1)y\big) \\[0.2cm] \frac{1}{2\nu}\big((\nu-1)x+(\nu+1)y\big)\\[0.2cm] 
  xy+\frac{\nu^2-1}{2}(x+y) \\[0.2cm]  1 \end{pmatrix}.
 \end{equation}

\section{Main theorem}
\label{sect main th}

\begin{theorem}\label{theorem main}
Consider a 3D QRT map $f$ defined by two pencils of quadrics $\cQ_\lambda$, $\cP_\mu$. On any affine part of $\mathbb P^3$, say $\{X_4\neq 0\}$, in affine coordinates 
$$
(x_1,x_2,x_3)=\Big(\frac{X_1}{X_4},\frac{X_2}{X_4},\frac{X_3}{X_4}\Big),
$$
the 3D QRT map preserves the volume form $dx_1\wedge dx_2\wedge dx_3/\rho(x_1,x_2,x_3)$ with the density
\begin{equation}\label{rho in aff coord}
\rho(x_1,x_2,x_3)=Q_{\infty}(x_1,x_2,x_3,1)P_{\infty}(x_1,x_2,x_3,1).
\end{equation}
Equivalently,
\begin{equation}\label{Jac dim 3 QRT}
\det\frac{\partial(\widetilde x_1,\widetilde x_2,\widetilde x_3)}{\partial(x_1,x_2,x_3)}=\frac{\rho(\widetilde x_1,\widetilde x_2,\widetilde x_3)}{\rho(x_1,x_2,x_3)},
\end{equation}
where we denote $[\widetilde x_1:\widetilde x_2:\widetilde x_3:1]=f([x_1:x_2:x_3:1])$.
\end{theorem}

A couple of comments are in order here. First of all, in this statement one can replace $Q_\infty$, $P_\infty$ by $Q_{\lambda_0}$, $P_{\mu_0}$ with arbitrary $\lambda_0, \mu_0\in\bbP^1$, since any two functions $Q_{\lambda_0}/Q_{\infty}$ and $P_{\mu_0}/P_{\infty}$ are integrals of motion for $f$. Second, though the statement of the theorem appears dependent on the choice of the affine part of $\mathbb P^3$, it actually is not, due to the following observation.

\begin{proposition}\label{prop Jac in dim 3,4}
Consider a rational map on $\mathbb P^3$ given by $\widetilde X_i=\widetilde X_i(X)$, $i=1,\ldots,4$, where $\widetilde X_i(X)$ are homogeneous polynomials of degree $m$. On any affine part of $\mathbb P^3$, say $\{X_4\neq 0\}$, in affine coordinates 
$$
(x_1,x_2,x_3)=\Big(\frac{X_1}{X_4},\frac{X_2}{X_4},\frac{X_3}{X_4}\Big),
$$
we have:
\begin{equation}\label{Jac in dim 3,4}
\det\frac{\partial (\widetilde X_1,\widetilde X_2,\widetilde X_3,\widetilde X_4)}{\partial(X_1,X_2,X_3,X_4)}=
m\frac{\widetilde X_4^4}{X_4^4}\det\frac{\partial(\widetilde x_1,\widetilde x_2,\widetilde x_3)}{\partial(x_1,x_2,x_3)}
\end{equation}
\end{proposition}
\begin{proof} The Jacobian on the left-hand side of equation \eqref{Jac in dim 3,4} can be transformed with the use of the Euler theorem for homogeneous polynomials (we write $\partial_i$ for $\partial/\partial X_i$, $i=1,\ldots,4$):
$$
\partial_4\widetilde X_i=-\frac{1}{X_4}\sum_{j=1}^3 X_j\partial_j \widetilde X_i+\frac{m\widetilde X_i}{X_4}, \quad i=1,2,3.
$$
Substitute this into the 4th column of the Jacobian and perform standard column operations. We come to 
$$
\frac{m}{X_4}\left|\begin{array}{cccc} 
\partial_1\widetilde X_1 & \partial_2\widetilde X_1 & \partial_3\widetilde X_1 & \widetilde X_1 \\
\partial_1\widetilde X_2 & \partial_2\widetilde X_2 & \partial_3\widetilde X_2 & \widetilde X_2 \\
\partial_1\widetilde X_3 & \partial_2\widetilde X_3 & \partial_3\widetilde X_3 & \widetilde X_3 \\
\partial_1\widetilde X_4 & \partial_2\widetilde X_4 & \partial_3\widetilde X_4 & \widetilde X_4 
\end{array}\right|.
$$
Divide the first three columns by $\widetilde X_4$ and subtract the 4th column multiplied by $(\partial_i\widetilde X_4)/\widetilde X_4^2$. The result reads
$$
\frac{m\widetilde X_4^3}{X_4}\left|\begin{array}{cccc} 
\frac{\partial_1\widetilde X_1}{\widetilde X_4}-\widetilde X_1\frac{\partial_1\widetilde X_4}{\widetilde X_4^2} & \frac{\partial_2\widetilde X_1}{\widetilde X_4}-\widetilde X_1\frac{\partial_2\widetilde X_4}{\widetilde X_4^2} & \frac{\partial_3\widetilde X_1}{\widetilde X_4}-\widetilde X_1\frac{\partial_3\widetilde X_4}{\widetilde X_4^2} & \widetilde X_1 
\\[0.2cm]
\frac{\partial_1\widetilde X_2}{\widetilde X_4}-\widetilde X_2\frac{\partial_1\widetilde X_4}{\widetilde X_4^2} & \frac{\partial_2\widetilde X_2}{\widetilde X_4}-\widetilde X_2\frac{\partial_2\widetilde X_4}{\widetilde X_4^2} & \frac{\partial_3\widetilde X_2}{\widetilde X_4}-\widetilde X_2\frac{\partial_3\widetilde X_4}{\widetilde X_4^2} & \widetilde X_1 
\\[0.2cm]
\frac{\partial_1\widetilde X_3}{\widetilde X_4}-\widetilde X_3\frac{\partial_1\widetilde X_4}{\widetilde X_4^2} & \frac{\partial_2\widetilde X_3}{\widetilde X_4}-\widetilde X_3\frac{\partial_2\widetilde X_4}{\widetilde X_4^2} & \frac{\partial_3\widetilde X_3}{\widetilde X_4}-\widetilde X_3\frac{\partial_3\widetilde X_4}{\widetilde X_4^2} & \widetilde X_1 \\
0 & 0 & 0 & \widetilde X_4 
\end{array}\right|,
$$
which equals
$$
\frac{m\widetilde X_4^4}{X_4}\left|\begin{array}{ccc} 
\partial_1\Big(\frac{\widetilde X_1}{\widetilde X_4}\Big) & \partial_2\Big(\frac{\widetilde X_1}{\widetilde X_4}\Big) & 
\partial_3\Big(\frac{\widetilde X_1}{\widetilde X_4}\Big) \\[0.2cm]
\partial_1\Big(\frac{\widetilde X_2}{\widetilde X_4}\Big) & \partial_2\Big(\frac{\widetilde X_2}{\widetilde X_4}\Big) & 
\partial_3\Big(\frac{\widetilde X_2}{\widetilde X_4}\Big) \\[0.2cm]
\partial_1\Big(\frac{\widetilde X_3}{\widetilde X_4}\Big) & \partial_2\Big(\frac{\widetilde X_3}{\widetilde X_4}\Big) & 
\partial_3\Big(\frac{\widetilde X_3}{\widetilde X_4}\Big) 
\end{array}\right| = \frac{m\widetilde X_4^4}{X_4}\left|\begin{array}{ccc} 
\partial_1\widetilde x_1 & \partial_2\widetilde x_1 &  \partial_3\widetilde x_1 \\[0.2cm]
\partial_1\widetilde x_2 & \partial_2\widetilde x_2 &  \partial_3\widetilde x_2 \\[0.2cm]
\partial_1\widetilde x_3 & \partial_2\widetilde x_3 &  \partial_3\widetilde x_3 
\end{array}\right| .
$$
Replacing $\partial_i=\partial/\partial X_i=X_4^{-1}\partial/\partial x_i$, we finish the proof. \end{proof}
\medskip

According to Proposition \ref{prop Jac in dim 3,4}, formula \eqref{Jac dim 3 QRT} takes the symmetric form
\begin{equation}\label{Jac dim 4 QRT}
\det\frac{\partial(\widetilde X_1,\widetilde X_2,\widetilde X_3,\widetilde X_4)}{\partial(X_1,X_2,X_3,X_4)}=
m\frac{Q_{\infty}(\widetilde X_1,\widetilde X_2,\widetilde X_3,\widetilde X_4)P_{\infty}(\widetilde X_1,\widetilde X_2,\widetilde X_3,\widetilde X_4)}{Q_{\infty}(X_1,X_2,X_3,X_4)P_{\infty}(X_1,X_2,X_3,X_4)},
\end{equation}
so that the statement of Theorem \ref{theorem main} does not depend on the choice of the affine chart.

\section{Proof of Theorem \ref{theorem main}}
\label{sect main proof}

For one particular class of pencils $\{\cQ_\lambda\}$, namely the pencils of quadrics through a skew quadrilateral (which can be normalized to $Q_\lambda(X)=X_1X_2-\lambda X_3X_4$), Theorem \ref{theorem main} was proved through a direct verification of formula \eqref{Jac dim 4 QRT} by symbolic computation in the master thesis \cite{H} supervised by the second author. Actually, it was this result that made us believe that Theorem \ref{theorem main} holds in general. However, a 
general proof by this method does not look feasible, because of a high complexity of 3D QRT maps in homogeneous coordinates. The proof given in the present paper uses the pencil-adapted coordinates.

In the pencil-adapted coordinates $(x,y,\lambda)$ on $\mathbb P^3$, the 3D QRT map restricted to the quadric $\cQ_\lambda$ is a 2D QRT map defined by the pencil of biquadratic curves $\cC_{\lambda,\mu}$ which are intersections of $\cQ_\lambda$ with the quadrics of the second pencil $\cP_\mu$. In other words, for the map $(x,y,\lambda)\mapsto(\widetilde x,\widetilde y,\widetilde \lambda)$ we have $\widetilde\lambda=\lambda$, and therefore, by Proposition \ref{2D QRT inv measure},
$$
\det\frac{\partial(\widetilde x,\widetilde y,\widetilde\lambda)}{\partial(x,y,\lambda)}=\frac{C_{\lambda,\infty}(\widetilde x,\widetilde y)}{C_{\lambda,\infty}(x,y)}.
$$
As we are interested in the expression of the density of the invariant measure in affine coordinates, say
$$
(x_1,x_2,x_3)=\Big(\frac{X_1}{X_4},\frac{X_2}{X_4},\frac{X_3}{X_4}\Big),
$$
we get:
$$
\det\frac{\partial(\widetilde x_1,\widetilde x_2,\widetilde x_3)}{\partial(x_1,x_2,x_3)}=
\det\frac{\partial(\widetilde x_1,\widetilde x_2,\widetilde x_3)}{\partial(\widetilde x,\widetilde y,\widetilde\lambda)}\cdot
\frac{C_{\lambda,\infty}(\widetilde x,\widetilde y)}{C_{\lambda,\infty}(x,y)}\cdot\bigg( \det\frac{\partial(x_1,x_2,x_3)}{\partial(x,y,\lambda)}\bigg)^{-1}.
$$
Thus, there exists an invariant measure with the density given by
\begin{equation}\label{rho in adapted coords}
\rho(x_1,x_2,x_3)=C_{\lambda,\infty}(x,y) \cdot \det\frac{\partial(x_1,x_2,x_3)}{\partial(x,y,\lambda)},
\end{equation}
and it remains to express the right-hand side in terms of $(x_1,x_2,x_3)$. 

For the first factor, we can use the following general formula:
\begin{equation}\label{C lambda mu}
C_{\lambda,\mu}(x,y)=P_\mu(\phi_\lambda(x,y)).
\end{equation}
For the second factor, we will prove:

\begin{proposition}\label{prop Q infty}
We have the following evaluation:
\begin{equation}\label{eq Q infty}
\det\frac{\partial(x_1,x_2,x_3)}{\partial(x,y,\lambda)}\simeq X_4^{-4} Q_\infty(\phi_\lambda(x,y)).
\end{equation}
Here, the symbol $\simeq$ means ``up to a constant factor, depending on $\lambda$ only''.
\end{proposition}

Formulas \eqref{rho in adapted coords}, \eqref{C lambda mu} and \eqref{eq Q infty} immediately imply \eqref{rho in aff coord}. \qed

\section{Examples}
\label{sect examples}

Before giving a general proof of Proposition \ref{prop Q infty}, we illustrate this result by the corresponding evaluations in our two examples.
\smallskip

\emph{Example 1: pencil of quadrics through a twisted cubic and one of its tangent lines.}

\noindent
We have the following parametrization of $\cQ_\lambda$ by the pencil-adapted coordinates in the affine patch $X_4=1$: 
\begin{equation}
\phi_\lambda(x,y)= \begin{pmatrix} x_1 \\ x_2\\ x_3 \\ 1\end{pmatrix}
 =\begin{pmatrix} x \\ y+\lambda x\\  xy+\lambda y+\lambda^2 x \\ 1\end{pmatrix}
 \end{equation}
 We compute:
  \begin{eqnarray*}
 \det\frac{\partial(x_1,x_2,x_3)}{\partial(x,y,\lambda)} & = & 
 \left|\begin{array}{ccc}
 1 & 0 & 0 \\
 \lambda & 1 & x \\
 y+\lambda^2 & x+\lambda & y+2\lambda x
 \end{array}\right| \\[0.2cm]
 & = & y+\lambda x-x^2=x_2-x_1^2=-Q_\infty(x_1,x_2,x_3,1).
\end{eqnarray*}

\emph{Example 2: pencil of quadrics through a cuspidal quartic curve.}

\noindent
We have the following parametrization of $\cQ_{\lambda(\nu)}$ by the pencil-adapted coordinates in the affine patch $X_4=1$: 
\begin{equation}
\phi_{\nu}(x,y)= \begin{pmatrix} x_1 \\ x_2\\ x_3 \\ 1 \end{pmatrix}
 =\begin{pmatrix} \frac{1}{2\nu}\big((\nu+1)x+(\nu-1)y\big) \\[0.2cm] \frac{1}{2\nu}\big((\nu-1)x+(\nu+1)y\big)\\[0.2cm] 
  xy+\frac{\nu^2-1}{2}(x+y) \\[0.2cm]  1 \end{pmatrix}.
 \end{equation}
 We compute:
 $$
 \det\frac{\partial(x_1,x_2,x_3)}{\partial(x,y,\nu)}=
 \left|\begin{array}{ccc}
 \frac{\nu+1}{2\nu} &  \frac{\nu-1}{2\nu} & -\frac{1}{2\nu^2}(x-y) \\[0.2cm]
 \frac{\nu-1}{2\nu} &  \frac{\nu+1}{2\nu} & \frac{1}{2\nu^2}(x-y) \\[0.2cm]
y+\frac{\nu^2-1}{2} & x+\frac{\nu^2-1}{2} & \nu(x+y)
 \end{array}\right|.
 $$
 A straightforward computation leads to
 \begin{eqnarray*}
 \det\frac{\partial(x_1,x_2,x_3)}{\partial(x,y,\nu)} & = &(x+y)-\frac{1}{2\nu^2}(x-y)^2\\
& = & (x_1+x_2)-\frac{1}{2}(x_1-x_2)^2\;=\; -Q_\infty(x_1,x_2,x_3,1).
\end{eqnarray*}

\section{Proof of Proposition \ref{prop Q infty}}
\label{sect proof prop 3}

To prove \eqref{eq Q infty}, we start with representing the left-hand side in a more symmetric form in homogeneous coordinates. Computations are very similar to those used in the proof of Proposition \ref{prop Jac in dim 3,4}. We have:
\begin{eqnarray*}
\det\frac{\partial(x_1,x_2,x_3)}{\partial(x,y,\lambda)} & = & 
\left|\begin{array}{ccc} 
\partial_x\Big(\frac{X_1}{X_4}\Big) & \partial_y\Big(\frac{X_1}{X_4}\Big) & \partial_\lambda\Big(\frac{X_1}{X_4}\Big) \\[0.2cm]
\partial_x\Big(\frac{X_2}{X_4}\Big) & \partial_y\Big(\frac{X_2}{X_4}\Big) & \partial_\lambda\Big(\frac{X_2}{X_4}\Big) \\[0.2cm]
\partial_x\Big(\frac{X_3}{X_4}\Big) & \partial_y\Big(\frac{X_3}{X_4}\Big) & \partial_\lambda\Big(\frac{X_3}{X_4}\Big) 
\end{array}\right|\\
& = & 
\left|\begin{array}{ccc} 
\frac{\partial_x X_1}{X_4}-X_1\frac{\partial_x X_4}{X_4^2} & \frac{\partial_y X_1}{X_4}-X_1\frac{\partial_y X_4}{X_4^2} & \frac{\partial_\lambda X_1}{X_4}- X_1\frac{\partial_\lambda X_4}{X_4^2}  
\\[0.2cm]
\frac{\partial_x X_2}{X_4}- X_2\frac{\partial_x X_4}{X_4^2} & \frac{\partial_y X_2}{X_4}- X_2\frac{\partial_y X_4}{X_4^2} & \frac{\partial_\lambda X_2}{X_4}- X_2\frac{\partial_\lambda X_4}{X_4^2} 
\\[0.2cm]
\frac{\partial_x X_3}{X_4}- X_3\frac{\partial_x X_4}{X_4^2} & \frac{\partial_y X_3}{X_4}- X_3\frac{\partial_y X_4}{X_4^2} & \frac{\partial_\lambda X_3}{X_4}- X_3\frac{\partial_\lambda X_4}{X_4^2} \end{array}\right|.
\end{eqnarray*}
This can be represented as
\begin{eqnarray*}
\det\frac{\partial(x_1,x_2,x_3)}{\partial(x,y,\lambda)} & = & \frac{1}{X_4}
\left|\begin{array}{cccc} 
\frac{\partial_x X_1}{X_4}-X_1\frac{\partial_x X_4}{X_4^2} & \frac{\partial_y X_1}{X_4}-X_1\frac{\partial_y X_4}{X_4^2} & \frac{\partial_\lambda X_1}{X_4}- X_1\frac{\partial_\lambda X_4}{X_4^2} & X_1 
\\[0.2cm]
\frac{\partial_x X_2}{X_4}- X_2\frac{\partial_x X_4}{X_4^2} & \frac{\partial_y X_2}{X_4}- X_2\frac{\partial_y X_4}{X_4^2} & \frac{\partial_\lambda X_2}{X_4}- X_2\frac{\partial_\lambda X_4}{X_4^2} & X_2 
\\[0.2cm]
\frac{\partial_x X_3}{X_4}- X_3\frac{\partial_x X_4}{X_4^2} & \frac{\partial_y X_3}{X_4}- X_3\frac{\partial_y X_4}{X_4^2} & \frac{\partial_\lambda X_3}{X_4}- X_3\frac{\partial_\lambda X_4}{X_4^2} & X_3
\\[0.2cm]
0 & 0 & 0 & X_4
\end{array}\right| \\
& = & \frac{1}{X_4^4}
\left|\begin{array}{cccc} 
\partial_x X_1 & \partial_y X_1 & \partial_\lambda X_1 & X_1 \\
\partial_x X_2 & \partial_y X_2 & \partial_\lambda X_2 & X_2 \\
\partial_x X_3 & \partial_y X_3 & \partial_\lambda X_3 & X_3 \\
\partial_x X_4 & \partial_y X_4 & \partial_\lambda X_4 & X_4
\end{array}\right|.
\end{eqnarray*}
Now formula \eqref{eq Q infty} reduces to  the following evaluation: for $X=\phi_\lambda(x,y)$, we have
\begin{equation}\label{eq Q infty hom}
\left|\begin{array}{cccc} 
\partial_x X_1 & \partial_y X_1 & \partial_\lambda X_1 & X_1 \\
\partial_x X_2 & \partial_y X_2 & \partial_\lambda X_2 & X_2 \\
\partial_x X_3 & \partial_y X_3 & \partial_\lambda X_3 & X_3 \\
\partial_x X_4 & \partial_y X_4 & \partial_\lambda X_4 & X_4
\end{array}\right| \simeq Q_\infty(\phi_\lambda(x,y)).
\end{equation}

\begin{proof}
Setting $Y=(x,y,xy,1)^{\rm T}$, so that $X=A_\lambda Y$, we have:
$$
\det(\partial_x X, \partial_y X, X, \partial_\lambda X)=\det(A_\lambda)\cdot\det(\partial_x Y,\partial_y Y, Y, A_\lambda^{-1}\partial_\lambda(A_\lambda) Y)
$$
Here, $\det(A_\lambda)$ is the scalar factor hidden in \eqref{eq Q infty hom} behind the symbol $\simeq$, and it will be ignored for the rest of this proof. Denoting for a moment $Z=A_\lambda^{-1}\partial_\lambda(A_\lambda) Y$, we compute:
$$
\det(\partial_x X, \partial_y X, X, \partial_\lambda X)\simeq 
\left|\begin{array}{cccc} 1 & 0 & x & Z_1 \\ 0 & 1 & y & Z_2 \\ y & x & xy & Z_3 \\ 0 & 0 & 1 & Z_4 \end{array}\right|.
$$
Developing the determinant with respect to the last column, we find:
$$
\det(\partial_x X, \partial_y X, X, \partial_\lambda X)\simeq 
y Z_1+x Z_2-Z_3-xyZ_4=Y^{\rm T}\begin{pmatrix} 0 & 1 & 0 & 0 \\ 1 & 0 & 0 & 0 \\ 0 & 0 & 0 & -1 \\ 0 & 0 & -1 & 0 \end{pmatrix} Z=Y^{\rm T}M_0 Z.
$$
Upon substituting the value $Z=A_\lambda^{-1}\partial_\lambda(A_\lambda) Y$, we find:
\begin{eqnarray*}
\det(\partial_x X, \partial_y X, X, \partial_\lambda X) & \simeq &  
Y^{\rm T} M_0 A_\lambda^{-1}\partial_\lambda(A_\lambda) Y\\
 & = & X^{\rm T} A_\lambda^{-\rm T}M_0A_\lambda^{-1}\partial_\lambda(A_\lambda) A_\lambda^{-1}X\\
 & = & X^{\rm T} M_\lambda\partial_\lambda(A_\lambda) A_\lambda^{-1}X.
\end{eqnarray*}
This is a quadratic form in $X$. Symmetrizing its matrix, we have (up to the factor 1/2, which can be still hidden behind $\simeq$):
\begin{equation}\label{aux}
\det(\partial_x X, \partial_y X, X, \partial_\lambda X) \simeq 
X^{\rm T} \Big(M_\lambda\partial_\lambda(A_\lambda) A_\lambda^{-1}+A_\lambda^{-\rm T}\partial_\lambda(A_\lambda^{\rm T})M_\lambda\Big)X.
\end{equation}
Now, differentiating 
$$
M_\lambda=A_\lambda^{-\rm T}M_0A_\lambda^{-1}
$$
with respect to $\lambda$, we find:
\begin{eqnarray*}
\partial_\lambda(M_\lambda) & = & -A_\lambda^{-\rm T}\partial_\lambda(A_\lambda^{\rm T})A_\lambda^{-\rm T}M_0A_\lambda^{-1}-
A_\lambda^{-\rm T}M_0A_\lambda^{-1}\partial_\lambda(A_\lambda)A_\lambda^{-1}\\
& = & -A_\lambda^{-\rm T}\partial_\lambda(A_\lambda^{\rm T})M_\lambda-
M_\lambda\partial_\lambda(A_\lambda)A_\lambda^{-1}.
\end{eqnarray*}
Comparing this with \eqref{aux}, we finally arrive at: 
\begin{equation}
\det(\partial_x X, \partial_y X, X, \partial_\lambda X) \simeq 
-X^{\rm T} \partial_\lambda(M_\lambda) X=X^{\rm T} M_\infty X=Q_\infty(X).
\end{equation}
This proves Proposition \ref{prop Q infty}, and with it also Theorem \ref{theorem main}.
\end{proof}

\section{Conclusions}
\label{sect conclusions}

With finding an invariant volume form for 3D QRT maps, the latter acquire the last attribute necessary to qualify as integrable systems by a traditional definition of integrability. This will be useful for the novel theory of discrete Painlev\'e equations which builds upon 3D QRT maps \cite{ASW1, AS2}. On the other hand, this result gives a unified approach to integrability of some old and new Kahan-Hirota-Kimura discretizations of 3D systems like Euler top and Zhukovsky-Volterra gyrostat, see \cite{PS, PPS}, since these discretizations have been identified as 3D QRT maps in \cite{ASW}.


\appendix
\section{Appendix: Proof of the invariant area form for 2D QRT maps}
Proposition \ref{2D QRT inv measure} follows from the anti-preservation of the area form $dx\wedge dy/C_\infty(x,y)$  by the both involutions $i_1$, $i_2$, that is,
\begin{equation}\label{App aux}
\det d i_1=\det\frac{\partial (x_0,y_1)}{\partial (x_0,y_0)}=-\frac{C_\infty(x_0,y_1)}{C_\infty(x_0,y_0)},
\end{equation}
and similarly for $i_2$. Here, $y_1$ is determined from $C_\mu(x_0,y_1)=C_\mu(x_0,y_0)$, where $\mu=\mu(x_0,y_0)$ is the unique value of the pencil parameter $\mu$ for which  $C_\mu(x_0,y_0)=0$. We write:
\begin{eqnarray*}
C_\mu(x,y) & = & c_0x^2y^2 + c_1x^2y + c_2xy^2 + c_3x^2 + c_4xy +c_5y^2 + c_6x + c_7y + c_8\\
 & = & \begin{pmatrix} x^2 & x & 1 \end{pmatrix} \begin{pmatrix} c_0 & c_1 & c_3 \\ c_2 & c_4 & c_6 \\ c_5 & c_7 & c_8 \end{pmatrix} \begin{pmatrix} y^2 \\ y \\ 1 \end{pmatrix} \; = \;\begin{pmatrix} x^2 & x & 1 \end{pmatrix} C_\mu\begin{pmatrix} y^2 \\ y \\ 1 \end{pmatrix}.
\end{eqnarray*}
We compute:
\begin{equation}\label{App aux4}
\det\frac{\partial (x_0,y_1)}{\partial (x_0,y_0)}=\frac{\partial y_1}{\partial y_0}.
\end{equation}
To compute the latter derivative, we differentiate $C_\mu(x_0,y_1)=0$ with respect to $y_0$, to obtain
$$
\begin{pmatrix} x_0^2 & x_0 & 1 \end{pmatrix} C_\mu\begin{pmatrix} 2y_1 \\ 1 \\ 0 \end{pmatrix}\frac{\partial y_1}{\partial y_0}+
\begin{pmatrix} x_0^2 & x_0 & 1 \end{pmatrix} \frac{\partial C_\mu}{\partial \mu}\begin{pmatrix} y_1^2 \\ y_1 \\ 1 \end{pmatrix}\frac{\partial \mu}{\partial y_0}=0,
$$
and thus
\begin{equation}\label{App aux5}
\frac{\partial y_1}{\partial y_0}=-\frac{\begin{pmatrix} x_0^2 & x_0 & 1 \end{pmatrix} \dfrac{\partial C_\mu}{\partial \mu}\begin{pmatrix} y_1^2 \\ y_1 \\ 1 \end{pmatrix}}{\begin{pmatrix} x_0^2 & x_0 & 1 \end{pmatrix} C_\mu\begin{pmatrix} 2y_1 \\ 1 \\ 0 \end{pmatrix}}\cdot\frac{\partial \mu}{\partial y_0}
\end{equation}
On the other hand, by differentiating $C_\mu(x_0,y_0)=0$ with respect to $y_0$, we obtain
$$
\begin{pmatrix} x_0^2 & x_0 & 1 \end{pmatrix} C_\mu\begin{pmatrix} 2y_0 \\ 1 \\ 0 \end{pmatrix}+
\begin{pmatrix} x_0^2 & x_0 & 1 \end{pmatrix} \frac{\partial C_\mu}{\partial \mu}\begin{pmatrix} y_0^2 \\ y_0 \\ 1 \end{pmatrix}\frac{\partial \mu}{\partial y_0}=0,
$$
and thus
\begin{equation}\label{App aux6}
\frac{\partial \mu}{\partial y_0}=-\frac{\begin{pmatrix} x_0^2 & x_0 & 1 \end{pmatrix} C_\mu\begin{pmatrix} 2y_0 \\ 1 \\ 0 \end{pmatrix}}{\begin{pmatrix} x_0^2 & x_0 & 1 \end{pmatrix} \dfrac{\partial C_\mu}{\partial \mu}\begin{pmatrix} y_0^2 \\ y_0 \\ 1 \end{pmatrix}}.
\end{equation}
Collecting everything from \eqref{App aux4}--\eqref{App aux6} together, we arrive at
\begin{equation}\label{App aux1}
\det\frac{\partial (x_0,y_1)}{\partial (x_0,y_0)}=\frac{\begin{pmatrix} x_0^2 & x_0 & 1 \end{pmatrix} \dfrac{\partial C_\mu}{\partial \mu}\begin{pmatrix} y_1^2 \\ y_1 \\ 1 \end{pmatrix}}{\begin{pmatrix} x_0^2 & x_0 & 1 \end{pmatrix} C_\mu\begin{pmatrix} 2y_1 \\ 1 \\ 0 \end{pmatrix}}\cdot \frac{\begin{pmatrix} x_0^2 & x_0 & 1 \end{pmatrix} C_\mu\begin{pmatrix} 2y_0 \\ 1 \\ 0 \end{pmatrix}}{\begin{pmatrix} x_0^2 & x_0 & 1 \end{pmatrix} \dfrac{\partial C_\mu}{\partial \mu}\begin{pmatrix} y_0^2 \\ y_0 \\ 1 \end{pmatrix}}.
\end{equation}
Now the key observation is the following : condition $C_\mu(x_0,y_1)=C_\mu(x_0,y_0)$ can be equivalently re-written as
$$
c_0x_0^2(y_0 + y_1) + c_1x_0^2 + c_2x_0(y_0 + y_1) + c_5(y_0 + y_1) + c_4 x_0 + c_7=0,
$$
or
$$
\begin{pmatrix} x_0^2 & x_0 & 1 \end{pmatrix} C_\mu\begin{pmatrix} y_0+y_1 \\ 1 \\  0\end{pmatrix}=0,
$$
which is equivalent to
\begin{equation}\label{App aux2}
\begin{pmatrix} x_0^2 & x_0 & 1 \end{pmatrix} C_\mu\begin{pmatrix} 2y_0 \\ 1 \\  0\end{pmatrix}=-\begin{pmatrix} x_0^2 & x_0 & 1 \end{pmatrix} C_\mu\begin{pmatrix} 2y_1 \\ 1 \\  0\end{pmatrix}.
\end{equation}
Upon using this in \eqref{App aux1}, we finally arrive at
\begin{equation}\label{App aux3}
\det\frac{\partial (x_0,y_1)}{\partial (x_0,y_0)}=-\frac{\begin{pmatrix} x_0^2 & x_0 & 1 \end{pmatrix} \dfrac{\partial C_\mu}{\partial \mu}\begin{pmatrix} y_1^2 \\ y_1 \\ 1 \end{pmatrix}}{\begin{pmatrix} x_0^2 & x_0 & 1 \end{pmatrix} \dfrac{\partial C_\mu}{\partial \mu}\begin{pmatrix} y_0^2 \\ y_0 \\ 1 \end{pmatrix}},
\end{equation}
which coincides with \eqref{App aux}, since $\partial C_\mu/\partial \mu=-C_\infty$. \qed


\end{document}